\documentclass[11pt]{article}

\usepackage[letterpaper,margin=1in]{geometry}

\usepackage{array}
\usepackage{tabularx}
\usepackage[table]{xcolor}
\usepackage{float}
\usepackage{graphicx}
\usepackage{amsmath}
\usepackage{amssymb}
\usepackage{amsfonts}
\usepackage{amsthm}
\usepackage{leftindex}
\usepackage{multicol}
\usepackage{afterpage}
\usepackage[hidelinks]{hyperref}
\usepackage[nameinlink]{cleveref}
\usepackage{placeins}
\usepackage{caption}
\usepackage{subcaption}
\usepackage[percent]{overpic}
\usepackage{makecell}
\usepackage{epstopdf}
\usepackage[ruled,vlined,linesnumbered]{algorithm2e}

\newcommand{\IR}{\mathbb{I}\mathbb{R}}

\theoremstyle{plain}
\newtheorem{theorem}{Theorem}
\newtheorem{corollary}{Corollary}

\theoremstyle{remark}

\newenvironment{keywords}{\par\noindent\textbf{Keywords:} }{\par}

\title{Computational Complexity Analysis of Interval Methods in Solving Uncertain Nonlinear Systems}

\author{%
Rudra Prakash\thanks{Department of Electrical Engineering, Indian Institute of Technology Delhi, New Delhi--110016, India. Emails: \texttt{rudra.prakash@ee.iitd.ac.in}, \texttt{janas@ee.iitd.ac.in}, \texttt{shaunak.sen@ee.iitd.ac.in}.}\textsuperscript{ $\,\dagger$}
\and
S.~Janardhanan\footnotemark[1]
\and
Shaunak Sen\footnotemark[1]%
}

\date{}

\hypersetup{
  pdftitle={Computational Complexity of Interval Methods},
  pdfauthor={Rudra Prakash, S. Janardhanan, and Shaunak Sen}
}

\begin{document}

\begingroup
\renewcommand{\thefootnote}{\fnsymbol{footnote}}
\maketitle
\footnotetext[2]{Corresponding author.}
\endgroup
\setcounter{footnote}{0}
\renewcommand{\thefootnote}{\arabic{footnote}}

\begin{abstract}
This paper analyzes the computational complexity of validated interval methods for uncertain nonlinear systems and steady-state enclosure. Interval analysis produces guaranteed enclosures that account for uncertainty and round-off, but its adoption is often limited by computational cost in high dimensions. We develop an algorithm-level worst-case framework that makes explicit the dependence on the problem dimension $n$, the initial search region size $\mathrm{Vol}(X_0)$, the target tolerance $\varepsilon$, and the costs of validated primitives (inclusion-function evaluation, Jacobian evaluation, and interval linear algebra).
Within this framework, we derive worst-case time and space bounds for interval bisection, subdivision$+$filter, interval constraint propagation, interval Newton, and interval Krawczyk, and identify dominant cost drivers. We also show that the computation of the determinant and inverse of interval matrices via naive Laplace expansion exhibits factorial growth with increasing matrix dimension, motivating specialized interval linear algebra.
We complement the worst-case bounds with computational results on two application-motivated biochemical steady-state models (a Hill-type regulatory network and an enzyme-saturation-based winner-take-all circuit) in dimensions $n\in\{2,5,10\}$, including instances that process millions of boxes. The resulting analysis and experiments support the practical design of validated solvers for uncertainty-aware steady-state screening tasks such as robust operating-point certification and multistability assessment.
\end{abstract}


\section{Introduction}
Mathematical modeling is central to understanding biomolecular systems in systems biology and synthetic biology~\cite{alon_introduction_2006,del_vecchio_biomolecular_2015}. A common description is through nonlinear ordinary differential equations (ODEs) that capture biochemical reaction dynamics. In many applications, qualitative and quantitative conclusions hinge on steady states and their dependence on parameters; consequently, considerable effort has been devoted to approximating or rigorously bounding steady-state solutions~\cite{mcbride_number_2020,streif_robustness_2016}.

A key obstacle is that nonlinear ODE models are rarely solvable in closed form, and numerical computation must contend with modeling uncertainty, parameter variability, and round-off errors. Interval analysis addresses these issues by replacing point estimates with intervals and propagating bounds through computations, thereby enabling validated conclusions~\cite{moore_introduction_2009}. Introduced by Ramon E.~Moore in the early 1960s, interval analysis has since matured into a broad toolkit for nonlinear equations, optimization, and differential equations~\cite{moore_introduction_2009,moore_methods_1979}. Its central objective is to produce guaranteed enclosures of the true solution while accounting for numerical uncertainties that conventional floating-point arithmetic alone cannot control~\cite{neumaier_interval_1990,tucker_validated_2023}.

Interval arithmetic operates on closed intervals and redefines operations (e.g., addition, multiplication, division) so that results contain all possible outcomes induced by the input ranges~\cite{moore_introduction_2009, jaulin_applied_2001, tucker_validated_2023}. This makes interval methods attractive for validated computations in domains including control engineering, robotics, chemical reaction network analysis, and systems and synthetic biology. In practice, steady-state enclosures are computed using a variety of methods, including interval bisection, subdivision$+$filter, interval constraint propagation, interval Newton, and interval Krawczyk~\cite{prakash_design_2024,chorasiya_quantitative_2023,jaulin_applied_2001,moore_introduction_2009,neumaier_interval_1990}.

Despite their rigor, interval methods can be computationally demanding, particularly for large-scale or highly nonlinear models. A principal driver is the curse of dimensionality: subdivision-based schemes may exhibit exponential growth in workload with the number of states or parameters~\cite{hansen_global_2004}. Additional overhead comes from evaluating interval extensions of nonlinear functions and their derivatives~\cite{moore_methods_1979}; interval arithmetic can be less efficient than floating-point arithmetic and may overestimate ranges due to dependency effects. Derivative-based strategies (e.g., interval Newton) and preconditioned variants (e.g., Krawczyk) add further cost through interval Jacobian evaluation and interval linear algebra, echoing classical cost drivers in large-scale nonlinear solvers (Newton/inexact Newton and nonlinear preconditioning)~\cite{eisenstat_walker_forcing_1996,cai_keyes_2002,cai_li_2011}. In safety-critical settings such as verified reachability~\cite{jaulin_applied_2001} and robust parameter estimation~\cite{neumaier_interval_1990}, understanding these costs is essential.

Computational complexity provides a principled lens for assessing practicality by relating time and memory requirements to an appropriate notion of input size~\cite{cormen_introduction_2022,arora_computational_2009}. For validated numerical algorithms, however, ``input size'' is not captured by dimension ($n$) alone: the cost also depends on the accuracy requirement (e.g., a tolerance $\varepsilon$), the size of the initial search region (e.g., $\mathrm{Vol}(X_0)$), and the costs of validated primitives (inclusion-function evaluation, Jacobian evaluation, and interval linear algebra). Interval methods are therefore compelling for safety-critical applications, but their performance cannot be characterised by a single, uniform complexity expression.

While the exponential worst-case behavior of subdivision (branch-and-bound) schemes is well known in general, existing results do not typically provide a unified, algorithm-level account for validated interval steady-state computation that makes explicit the dependence on the problem dimension $n$, the initial search-region size $\mathrm{Vol}(X_0)$, the target tolerance $\varepsilon$, and primitive costs. This paper addresses this gap.

\noindent\textbf{Contributions.} We derive worst-case time and space complexity bounds for representative interval methods used to compute rigorous enclosures of steady states. The bounds are stated in terms of the problem dimension $n$, $\mathrm{Vol}(X_0)$, $\varepsilon$, and the computational costs of validated primitives, enabling direct comparison of subdivision-based and derivative-based strategies (interval bisection, subdivision$+$filter, interval constraint propagation, interval Newton, and interval Krawczyk) and isolating dominant scalability bottlenecks. We also include numerical experiments on two application-motivated biochemical steady-state models that compare observed box counts and wall-clock runtimes across methods in dimensions $n\in\{2,5,10\}$, including instances that process up to $8.76\times 10^7$ boxes. These results are worst-case guarantees for the specific algorithms analyzed and are not claimed to be information-theoretic lower bounds for validated steady-state enclosure problems.

\noindent\textbf{Relevance to applications.} Validated steady-state enclosures arise directly in scientific and engineering workflows where qualitative behavior depends on equilibria under uncertainty. Beyond systems and synthetic biology (e.g., certifying multistability in gene circuits~\cite{gardner_construction_2000,elowitz_synthetic_2000}, bounding operating points in biochemical networks, and supporting robust parameter identification with interval-valued kinetics), the same validated root-finding primitives appear in robust control and safety verification (certified equilibrium regions and invariant sets~\cite{blanchini_set_1999,rakovic_invariant_2005}), robotics and localization (guaranteed pose/parameter estimation under bounded sensing errors~\cite{jaulin_interval_2000,kieffer_robust_2000}), and chemical/process engineering (steady-state and feasibility analysis of nonlinear reactor and flowsheet models with uncertain parameters~\cite{gau2002new,lin_lp_2004}). In such settings, the algorithmic choice is constrained not only by correctness guarantees but also by the wall-clock time required to process large numbers of boxes and to perform validated Jacobian and linear-algebra steps.

This paper investigates the worst-case computational time and space complexity of interval methods used to obtain rigorous enclosures of steady-state solutions, formulated as root-finding problems of the form
\begin{equation}
    f(x,u)=0, \quad x\in X_0, \quad u\in U.
\end{equation}
Here $x$ denotes the state variables and $u$ the uncertain parameters. The box $X_0$ is an initial search enclosure, and $U$ is an interval set specifying admissible parameter values.

Section~2 reviews interval analysis concepts and introduces the cost model and complexity notation used throughout. Section~3 derives worst-case time and space bounds for five representative interval methods, with an explicit accounting of primitive costs (interval evaluation, contraction, and interval linear algebra). Section~3 also reports numerical experiments that illustrate how observed workloads can deviate from conservative worst-case scaling. Section~4 concludes with implications and directions for scalable validated solvers.

\section{Background}
\subsection{Interval Analysis}
A key advantage of interval analysis is its ability to represent and propagate uncertainty. Unlike conventional floating-point computation, interval methods aim to produce \,\emph{validated} enclosures: the true (exact) quantity of interest is guaranteed to lie within the computed interval bounds.

At its core, interval analysis works with closed real intervals. Let $A=[\underline{a},\overline{a}]$ and $B=[\underline{b},\overline{b}]$ with endpoints in $\mathbb{R}$, and write $A,B\in\IR$. The elementary operations are lifted to intervals so that, for each operation $\circ\in\{+,-,\times,\div\}$, the interval result encloses all values $a\circ b$ with $a\in A$ and $b\in B$~\cite{jaulin_applied_2001}. When dividing by an interval that contains zero, extended interval arithmetic is required. An $n$-dimensional interval (box) is the Cartesian product $X=X_1\times\cdots\times X_n\in\IR^n$.

Interval analysis also provides interval extensions of functions, enabling guaranteed range enclosures over boxes. The following theorems summarize two central properties: inclusion isotonicity and range enclosure.
\begin{theorem}[\cite{moore_introduction_2009}]
Given a real-valued rational function $f$ and a natural interval extension $F$ such that $F(X)$ is well defined for some $X \in \IR^n$, we have:
\begin{enumerate}
    \item If $Y \subseteq Z \subseteq X$, then $F(Y) \subseteq F(Z)$ (\textit{inclusion isotonicity}).
    \item The range of $f$ over $X$ satisfies $R(f;X) \subseteq F(X)$ (\textit{range enclosure}).
\end{enumerate}
\end{theorem}
Despite its strengths, interval analysis presents challenges, primarily due to overestimation caused by the dependency problem, which arises when the same variable appears multiple times in an expression. A standard solution is subdivision: by evaluating $F$ on smaller sub-intervals and combining the results, the enclosure can be made arbitrarily sharp.
\begin{theorem}[\cite{moore_introduction_2009}]
Consider $f:I\rightarrow \mathbb{R}$, where $f$ is Lipschitz, and let $F$ be an inclusion isotonic interval extension of $f$ such that $F(X)$ is well defined for some $X\subseteq I$. Then there exists a constant $k>0$ (depending on $F$ and $X$) with the following property: if $X$ is partitioned into sub-intervals $X_1,\ldots,X_m$ and we denote their interval sum by $\sum_{i=1}^m X_i$, then
\[
R(f;X) \subseteq \sum_{i=1}^{m} F(X_i) \subseteq F(X)
\]
and the radius satisfies
\[
\operatorname{rad}\Bigl(\sum_{i=1}^{m} F(X_i)\Bigr) \le \operatorname{rad}(R(f;X)) + k\,\max_{i=1,\ldots,m}\operatorname{rad}(X_i).
\]
\end{theorem}

Because the presence of multiple solutions, overestimated bounds, and the dependency issues described above often require using finer subdivisions to maintain accuracy, it is essential to understand the computational complexity of these operations in order to evaluate whether they are feasible in practical, real-world systems.

\subsection{Computational Complexity}
Time and space complexity quantify, respectively, the number of elementary operations and the amount of memory required by an algorithm as a function of an appropriate notion of input size~\cite{cormen_introduction_2022}. Throughout this paper, we use Big-$O$ notation to state worst-case upper bounds, counting primitive operations such as interval arithmetic, inclusion-function evaluations, Jacobian evaluations, and interval linear algebra steps~\cite{cormen_introduction_2022}.

For validated numerical algorithms, ``input size'' is not captured by dimension alone. In addition to the state dimension $n$, the computational effort typically depends on an accuracy requirement (e.g., a tolerance $\varepsilon$), the size of the initial search region (e.g., $\mathrm{Vol}(X_0)$ or related width/diameter measures), and the cost of the underlying validated primitives (e.g., the cost $C_F$ of evaluating an inclusion function). Space complexity is driven by the storage of interval boxes, trees/queues generated by subdivision, and intermediate interval vectors and matrices. Making these dependencies explicit is essential for assessing scalability and for identifying the algorithmic bottlenecks that limit the applicability of interval methods in high-dimensional, safety-critical settings.

\section{Results and Discussion}
Building on the background and notation above, we now derive worst-case time and space bounds for representative interval methods. We begin by fixing a simple cost model for interval arithmetic and for the validated primitives used by each method.

\paragraph{Cost assumptions and symbols} We state worst-case bounds in Big-$O$ form by treating a scalar real arithmetic operation as constant cost, with unit cost denoted by $t$. We account for (i) scalar interval arithmetic operations, (ii) evaluations of an inclusion function $F(X)$, (iii) evaluations of an interval Jacobian $J(X)$, and (iv) linear-algebra steps used within a validated computation.
Specifically, $C_F$ denotes the cost of one inclusion-function evaluation $F(X)$, $C_J$ denotes the cost of computing an interval Jacobian $J(X)$, $C_{J^{-1}}$ denotes the cost of computing (an enclosure of) the inverse of an interval Jacobian, and $N_{\mathrm{it}}$ denotes the maximum number of Newton/Krawczyk iterations performed per sub-box.
Unless stated otherwise, the total cost of processing a box is given by a combination of these primitives and scales with the number of sub-boxes needed to reach the tolerance $\varepsilon$.

\begin{table}[h!]
\centering
\captionsetup{font=normalsize}
\caption{Time complexity comparison: interval arithmetic vs. real arithmetic.}

\begingroup
\small\selectfont
\setlength{\tabcolsep}{4pt}
\renewcommand{\arraystretch}{0.95}
\begin{tabular}{|>{\centering\arraybackslash}m{4.5cm}|>{\centering\arraybackslash}m{1.5cm}|>{\centering\arraybackslash}m{2.3cm}|>{\centering\arraybackslash}m{1.5cm}|>{\centering\arraybackslash}m{2.3cm}|}
\hline
\rowcolor{lightgray!25}
\textbf{Operation} & \textbf{Scalar Real} & \textbf{\(n\)-Dimensional Real Vector Cost} & \textbf{Scalar Interval} & \textbf{$n$-Dimensional Interval Vector} \\\hline
Addition $(A + B)$ & $O(t)$ & $O(tn)$ & $O(2t)$ & $O(2tn)$ \\\hline
Subtraction $(A - B)$ & $O(t)$ & $O(tn)$ & $O(2t)$ & $O(2tn)$ \\\hline
Multiplication $(A \times B)$ & $O(t)$ & $O(tn)$ & $O(10t)$ & $O(10tn)$ \\\hline
Division $(A \div B; 0 \notin B)$ & $O(t)$ & $O(tn)$ & $O(12t)$ & $O(12tn)$  \\\hline
\end{tabular}
\endgroup

\label{table:int_arith_cost}
\end{table}
In Table~\ref{table:int_arith_cost}, the notation $O(t)$ represents the cost of one scalar real arithmetic operation. The constant factors shown for the interval operations come from counting the underlying real-endpoint arithmetic; the corresponding operation-count derivations are provided in Appendix~\ref{app:interval_arithmetic}. These constants are implementation-dependent; throughout the remainder we primarily track the asymptotic scaling in $n$, $\varepsilon$, and $\mathrm{Vol}(X_0)$.

\subsection{Interval Bisection Method}
Interval bisection is the validated analogue of classical bisection for enclosing the roots of a nonlinear system. Starting from an initial box $X_0$, the method recursively splits boxes and retains only those for which the interval evaluation cannot exclude a root (i.e., $0\in F(X,U)$). The recursion terminates when the box diameter is at most $\varepsilon$. Algorithm~\ref{algo:recursive_bisection} summarizes the method~\cite{prakash_design_2024,prakash_rigorous_2025}.

\begin{algorithm}[htbp]
    \caption{Recursive Interval Bisection}
    \label{algo:recursive_bisection}

    \DontPrintSemicolon

    \KwIn{$F$, $X_0$, $U$, $\varepsilon$.}
    \KwOut{Set}

    \SetKwFunction{IntervalBisection}{IntervalBisection}
    \SetKwProg{Fn}{Function}{}{}

    \Fn{\IntervalBisection$(X)$}
    {
        \If{$0 \in F(X, U)$}
        {
            \eIf{$\operatorname{diam}(X) \leq \varepsilon$}
            {
                \text{Save $X$ in Set}\;
            }
            {
                \text{Bisect $X$ into $X_{left}$ and $X_{right}$}\;
                $\IntervalBisection(X_{left})$\;
                $\IntervalBisection(X_{right})$\;
            }
        }
        \KwRet{$Set$}\;
    }

    ${\IntervalBisection}(X_0)$
\end{algorithm}

The initial search space is either proven not to contain any roots (in which case it is discarded) or it is bisected and kept for further study. The output is a collection of sub-boxes whose union encloses all roots of $f$ within $X_0$ (validated via the inclusion test $0\in F(X,U)$).

\paragraph{Assumptions and enclosure goal for the complexity analysis} We consider the steady-state (root-finding) problem $f(x,u)=0$ with $x\in X_0$ and $u\in U$ (as introduced in the Introduction), and focus on validated \emph{enclosure of all} steady states in $X_0$ (rather than existence/uniqueness of a single solution). We assume that $f$ is continuously differentiable in $x$ on the region of interest so that the Jacobian exists and can be bounded. Subdivision-based methods terminate when every retained sub-box has width (diameter) at most $\varepsilon$ in each coordinate; derivative-based methods may additionally terminate early on a sub-box when a contraction/verification test succeeds or fails, but a maximum of $N_{\mathrm{it}}$ iterations is enforced per sub-box.

\paragraph{Notation} $X_0$ denotes an $n$-dimensional search box (interval box), $X_0 = X_1 \times X_2 \times \ldots \times X_n \in \IR^n$. Let $\mathrm{Vol}(X_0)$ denote its volume, which we use as a measure of the size of the initial search region. The symbol $\varepsilon$ denotes the prescribed tolerance specifying the desired precision of the final interval boxes. The quantity $w(X_0)$ denotes the maximum width of its coordinate intervals (i.e., $\max_{i=1,\dots,n}(\overline{X_i}-\underline{X_i})$). In the context of interval boxes, the terms ``width'' and ``diameter'' are used synonymously. Finally, $n$ denotes the state dimension (number of variables).
\begin{theorem}\label{thm:CF_cost}
Let $f : \mathbb{R}^n \to \mathbb{R}^n$ be specified by an expression tree consisting of $k$ elementary scalar operations/functions (e.g., $+,-,\times,\div,\sin,\exp$). Let $F : \IR^n \to \IR^n$ be the natural interval extension obtained by replacing each elementary scalar operation by its interval counterpart.
Assume each scalar interval operation has worst-case cost $O(c)$ under the chosen cost model.
Then evaluating $F(X)$ on a box $X\in\IR^n$ has worst-case cost $T_F = O(ck)$.
\end{theorem}
\begin{proof}
By construction, $F$ executes the same sequence of $k$ elementary operations as $f$, with each real operation replaced by the corresponding interval operation. Under the assumption that each scalar interval operation costs $O(c)$, the total evaluation cost is at most $k\cdot O(c)=O(ck)$.
\end{proof}

Throughout the paper, we use the shorthand
\(\label{eqn:C_F}
C_F := O(ck).
\)

\textit{Example:} Consider the mapping \( f : \mathbb{R}^2 \to \mathbb{R}^2 \) defined by
\[
f(x_1, x_2) = \bigl(f_1(x_1, x_2), \ f_2(x_1, x_2)\bigr) = \bigl(x_1 + x_2,\; x_1(1 + x_2)\bigr),
\]
and let the domain (interval box) be
\(
X_0 = [1, 2] \times [3, 4].
\)
The image of \(X_0\) under \(f\) can be represented as
\[
F(X_0) = \bigl(F_1(X_0), \ F_2(X_0)\bigr),
\]
where interval arithmetic is used to compute the range of each component over \(X_0\).

For the first component,
\(
F_1(X_0) = [1, 2] + [3, 4] = [4, 6].
\)
For the second component,
\(
F_2(X_0) = [1, 2]\bigl([1, 1] + [3, 4]\bigr) = [1, 2] [4, 5] = [4, 10].
\)
Thus, the interval enclosure of the image of \(X_0\) under \(f\) is
\[
F(X_0) = \bigl(F_1(X_0), \ F_2(X_0)\bigr) = \bigl([4, 6], \ [4, 10]\bigr).
\]
The construction of the interval extension \(F\) of a real-valued function \(f\) involves three scalar interval-arithmetic operations in total: two interval additions and one interval multiplication. Consequently, the total computational cost can be expressed as \(C_F = O(3)O(c) = O(3c)\), where \(O(c)\) denotes the assumed maximum cost of a single scalar interval operation.

In this simple example, we can explicitly count the underlying real-valued operations: the two interval additions together require four real operations, while the interval multiplication requires ten real operations; see Table~\ref{table:int_arith_cost}. For more complicated representations of \(f\), however, it is not practical to enumerate and classify all interval operations and their specific types. For this reason, we adopt the simplifying assumption that each scalar interval operation incurs a cost of \(O(c)\).

\begin{corollary}\label{cor:diam_cost}
Let $X = X_1 \times \cdots \times X_n \in \IR^n$ denote an $n$-dimensional interval box, where each component interval is given by $X_i = [\underline{X}_i, \overline{X}_i]$. Define,
\[
\operatorname{diam}(X) := \max_{1 \le i \le n} \operatorname{diam}(X_i), \qquad \operatorname{diam}(X_i) := \overline{X}_i - \underline{X}_i.
\]
Under the assumption that scalar arithmetic operations and comparisons have unit computational cost, the value $\operatorname{diam}(X)$ can be evaluated in $O(n)$ time.
\end{corollary}
\begin{proof}
Compute $\operatorname{diam}(X_i)$ for $i=1,\dots,n$ using $n$ scalar subtractions, and then compute their maximum using $n-1$ comparisons. The total number of elementary operations is therefore linear in $n$, i.e., $O(n)$.
\end{proof}
\subsubsection{Time Complexity}
\begin{theorem}\label{thm:bisection_tc}
Under uniform subdivision until the coordinate-wise widths are $\leq\varepsilon$ (so each terminal box has volume $\leq\varepsilon^n$), the worst-case computational time complexity of interval bisection (Algorithm~\ref{algo:recursive_bisection}) for enclosing all roots of \( f(x,u) = 0 \) with \(x \in X_0\) and \(u \in U\) is
\[
O\!\left( (C_F+n)\,\frac{\mathrm{Vol}(X_0)}{\varepsilon^n} \right).
\]
\end{theorem}
\begin{proof}
In the worst case, the procedure continues subdividing until every terminal (leaf) box has coordinate-wise widths at most $\varepsilon$. Such a leaf box has volume at most $\varepsilon^n$, so the number of leaf boxes is at most
\[
L \le 
\left\lceil 
\frac{\mathrm{Vol}(X_0)}{\varepsilon^n} 
\right\rceil 
= O\!\left(\frac{\mathrm{Vol}(X_0)}{\varepsilon^n}\right).
\]
\begin{figure}[t]
\centering
\begin{subfigure}[t]{0.35\linewidth}
    \centering
    \begin{overpic}[width=\linewidth]{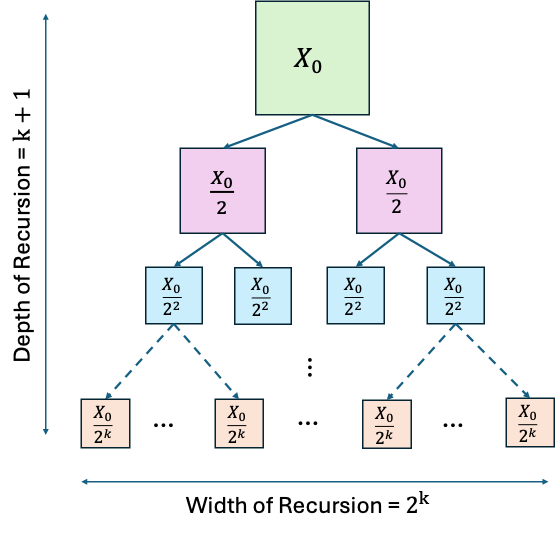}
        \put(0,97){\small\textbf{(a)}}
    \end{overpic}
    \phantomsubcaption\label{fig:recursion_tree}
\end{subfigure}\hfill
\begin{subfigure}[t]{0.48\linewidth}
    \centering
    \begin{overpic}[width=\linewidth]{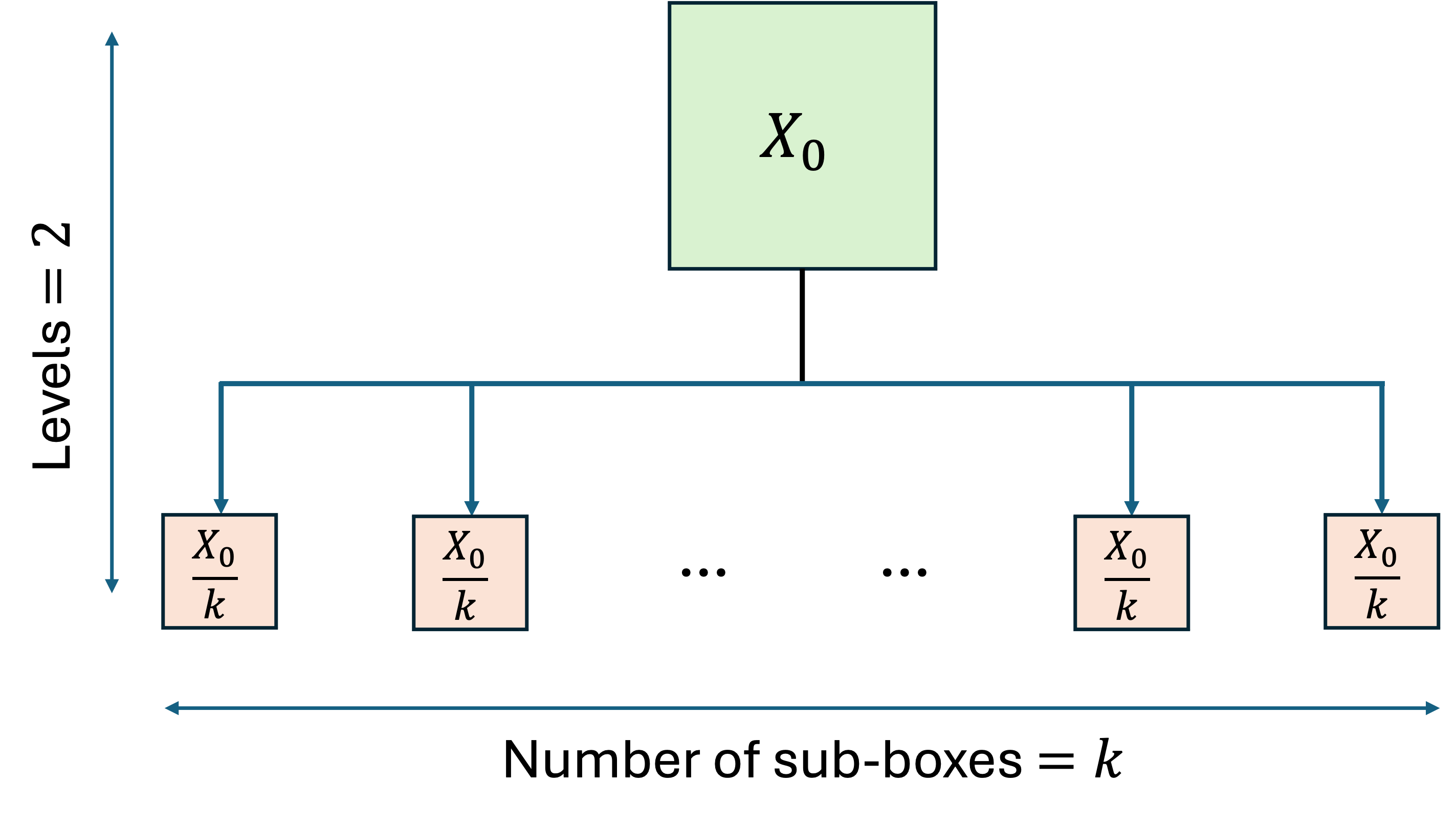}
        \put(0,70){\small\textbf{(b)}}
    \end{overpic}
    \phantomsubcaption\label{fig:subdivision_filter}
\end{subfigure}
\vspace{-0.5cm}
\caption{(a) Interval-bisection recursion tree; (b) uniform subdivision of the initial search box $X_0$ into $k$ sub-boxes used in subdivision$+$filter~\cite{prakash_design_2024}.}
\label{fig:subdivision_illustrations}
\end{figure}
A binary subdivision tree with $L$ leaves has $O(L)$ total nodes (and hence $O(L)$ recursive calls); see Fig.~\ref{fig:subdivision_illustrations}(a).

Each call performs one inclusion-function evaluation (cost $C_F$) and one diameter/width computation for termination tests (cost $O(n)$). Therefore the per-call cost is $O(C_F+n)$, and multiplying by $O(L)$ yields
\[
O\!\left((C_F+n)\,\frac{\mathrm{Vol}(X_0)}{\varepsilon^n}\right).
\]
\end{proof}
\subsubsection{Space Complexity}
\begin{theorem}\label{thm:bisection_sp}
The worst-case space complexity of interval bisection (Algorithm~\ref{algo:recursive_bisection}) is
\[
O\!\left(n\,\frac{\mathrm{Vol}(X_0)}{\varepsilon^n}\right)
\]
under uniform subdivision until the coordinate-wise widths are $\leq\varepsilon$.
\end{theorem}
\begin{proof}
The dominant memory cost is storing boxes either in the recursion stack/worklist or in the output set. Under uniform subdivision to resolution $\varepsilon$, the total number of boxes can be as large as $O\!\left(\mathrm{Vol}(X_0)/\varepsilon^n\right)$, and storing each box requires $O(n)$ endpoints. Up to constant factors (and lower-order stack-depth terms), this yields the stated bound.
\end{proof}
In practical implementations, the memory footprint is often smaller than the worst-case bound above, since many boxes are discarded early once an exclusion test succeeds.
More generally, the exponential dependence on $n$ is an inherent limitation of uniform subdivision schemes, and motivates aggressive pruning and adaptive subdivision strategies in higher dimensions.

A drawback of a purely recursive implementation is that it incurs recursion/stack overhead and may perform redundant evaluations on boxes that could instead be handled in a single pass over a worklist. A common alternative is therefore to first subdivide $X_0$ into a finite collection of boxes and then apply a cheap filtering step.
This leads to the following subdivision-and-filter approach.

\subsection{Subdivision \texorpdfstring{$+$}{+} Filter Method}
This method avoids recursion by processing a uniformly subdivided grid of boxes using an explicit worklist. This primarily reduces implementation overhead and makes memory usage more predictable; the dominant cost still scales with the number of sub-boxes created. The initial search space is partitioned into smaller subregions, and subregions that meet a non-existence criterion are filtered out~\cite{prakash_design_2024}. The remaining subregions are those that potentially contain a solution.
Figure~\ref{fig:subdivision_illustrations}(b) illustrates the uniform subdivision step applied to $X_0$; Algorithm~\ref{algo:subdivision_filter} summarizes the method~\cite{prakash_design_2024}.

\begin{algorithm}[h!]
    \caption{Subdivision $+$ Filter}
    \label{algo:subdivision_filter}

    \DontPrintSemicolon

    \KwIn{$F$, $X_0$, $U$, $m$.}
    \KwOut{Set}

    \SetKwFunction{SubdivisionFilter}{SubdivisionFilter}
    \SetKwProg{Fn}{Function}{}{}

    \Fn{\SubdivisionFilter$(X)$}
    {
        \If{$0 \in F(X, U)$}
        {
            \text{$X_m$ $\leftarrow$ Divide $X$ into $m$ parts per dimension.}\;
            \text{$Set$ $\leftarrow$ Select all $X_f$ in $X_m$ satisfying $0 \in F(X_f, U)$.}\;
        }
        \KwRet{$Set$}\;
    }

    ${\SubdivisionFilter}(X_0)$
\end{algorithm}

\subsubsection{Time Complexity}
\begin{theorem}\label{thm:subfilter_tc}
Let $m$ denote the number of subdivisions per dimension (so the grid contains $m^n$ sub-boxes). The worst-case computational time complexity of the Subdivision$+$Filter method (Algorithm~\ref{algo:subdivision_filter}) is
\(
O\!\left(m^n\,C_F\right).
\)
\end{theorem}
\begin{proof}
The algorithm evaluates the inclusion function once per sub-box in the $m^n$-box grid and retains those boxes $X$ for which $0\in F(X,U)$. Each inclusion-function evaluation costs $C_F$, hence the overall worst-case time is $O(m^n C_F)$.
\end{proof}
\noindent\textit{Remark.} Under uniform subdivision to a target coordinate-wise width $\varepsilon$, one may take $m\approx w(X_0)/\varepsilon$, which yields the familiar scaling $O\!\left(C_F\,(w(X_0)/\varepsilon)^n\right)$ (up to constants depending on the splitting convention).
\subsubsection{Space Complexity}
\begin{theorem}\label{thm:sub_filter_sp}
The worst-case computational space complexity of the Subdivision $+$ Filter method (Algorithm~\ref{algo:subdivision_filter}) is $O\!\left(m^n\right)$.
\end{theorem}
\begin{proof}
In the worst case, none of the $m^n$ sub-boxes are filtered out, and hence the method may need to store $O(m^n)$ boxes (either in the output set or in a worklist).
\end{proof}
Interval bisection and subdivision$+$filter are straightforward to implement for enclosing multiple roots of $f(x,u)=0$, but they can be expensive in high dimensions because they rely primarily on subdivision to reach the target tolerance. Moreover, when multiple steady states lie close together, a very small $\varepsilon$ may be required to separate them into distinct boxes, further increasing the number of boxes that must be processed.

\subsection{Interval Constraint Propagation}
Interval constraint propagation is a standard set-inversion approach for nonlinear constraint satisfaction with continuous variables~\cite{jaulin_applied_2001}. Given an initial box $X_0$, interval constraint propagation iteratively contracts candidate boxes using a contractor associated with the constraints and subdivides as needed to separate solutions. In the interval-analysis setting, this corresponds to the SIVIA framework~\cite{jaulin_applied_2001}. Algorithm~\ref{algo:contractor} summarizes the variant analyzed in this paper~\cite{prakash_design_2024, prakash_rigorous_2025}.

\begin{algorithm}[H]
    \caption{Interval Constraint Propagation}
    \label{algo:contractor}
    \DontPrintSemicolon
    \KwIn{Contractor, $X_0$, $\varepsilon$, $m$, $N_{\mathrm{it}}$}
    \KwOut{Set}
    \SetKwFunction{IntervalConstraintPropagation}{IntervalConstraintPropagation}
    \SetKwFunction{Contractor}{Contractor}
    \SetKwProg{Fn}{Function}{}{}

    \Fn{\IntervalConstraintPropagation{$X$}}{
        $X_m \leftarrow$ Divide $X$ into $m$ parts in each dimension\;
        \For{$X_i$ $\in$ $X_m$}
        {
            $X_{old} \leftarrow X_i$ \;
            \For{$j = 1,\ldots, N_{\mathrm{it}}$}
            {
                $X_{new} \leftarrow$ \Contractor($X_{old}$) \;
                \If{$\operatorname{diam}(X_{old}) - \operatorname{diam}(X_{new}) \leq \varepsilon / 10$ OR $X_{new} = \emptyset$}
                {
                    \textbf{break}
                }
                $X_{old} \leftarrow X_{new}$\;
            }
            \If{$X_{new}\neq \emptyset$}
            {
                Save $X_{new}$ in Set
            }
        }
        \KwRet{Set}\;
    }

${\IntervalConstraintPropagation}(X_0)$
\end{algorithm}

Additional notation: $m$ denotes the number of subdivisions per dimension used to form $X_m$ (so the total number of elements in $X_m$ is $m^n$), and $N_{\mathrm{it}}$ denotes the maximum number of contractor iterations performed per sub-box.
\subsubsection{Time Complexity}
\begin{theorem}\label{thm:icp_tc}
Let $C_{\mathrm{con}}$ denote the cost of one call to the contractor on an $n$-dimensional box. The worst-case computational time complexity of the interval constraint propagation method (Algorithm~\ref{algo:contractor}) for enclosing all zeros of \(f(x,u)=0\) with \(x\in X_0\) and \(u\in U\) is
\(
O\!\left(m^n\,N_{\mathrm{it}}\,\bigl(C_{\mathrm{con}}+n\bigr)\right).
\)
If one contractor call has the same asymptotic cost as evaluating the inclusion function once, i.e., $C_{\mathrm{con}}=O(C_F)$, then
\(
O\!\left(m^n\,N_{\mathrm{it}}\,(C_{\mathrm{con}}+n)\right)=O\!\left(m^n\,N_{\mathrm{it}}\,C_F\right)
\)
(up to lower-order diameter-check terms).
In particular, under uniform subdivision such that $m^n\approx \mathrm{Vol}(X_0)/\varepsilon^n$, this yields
\[
O\!\left(N_{\mathrm{it}}\,C_F\,\frac{\mathrm{Vol}(X_0)}{\varepsilon^n}\right).
\]
\end{theorem}
\begin{proof}
The algorithm partitions $X$ into $m^n$ sub-boxes. For each sub-box $X_i$, it performs at most $N_{\mathrm{it}}$ contractor iterations. Each iteration calls the contractor once (cost $C_{\mathrm{con}}$) and evaluates $\operatorname{diam}(\cdot)$ on $X_{\mathrm{new}}$ and $X_{\mathrm{old}}$ to test the stopping criterion, which costs $O(n)$ per iteration. Hence the worst-case cost per sub-box is $O\bigl(N_{\mathrm{it}}(C_{\mathrm{con}}+n)\bigr)$, and multiplying by $m^n$ sub-boxes gives the first bound. The second bound follows by substituting $m^n\approx \mathrm{Vol}(X_0)/\varepsilon^n$ under uniform subdivision to tolerance.
\end{proof}
\subsubsection{Space Complexity}
\begin{theorem}\label{thm:icp_sp}
The worst-case computational space complexity of the interval constraint propagation method (Algorithm~\ref{algo:contractor}) for enclosing all zeros of \(f(x,u)=0\) with \(x\in X_0\) and \(u\in U\) is
\(
O\!\left(m^n\right).
\)
In particular, under uniform subdivision such that $m^n\approx \mathrm{Vol}(X_0)/\varepsilon^n$, this yields $O\!\left(\mathrm{Vol}(X_0)/\varepsilon^n\right)$.
\end{theorem}
\begin{proof}
The algorithm processes sub-boxes sequentially and only needs to store a constant number of boxes during contraction (e.g., $X_{\mathrm{old}}$ and $X_{\mathrm{new}}$). The dominant memory cost is storing the output set of nonempty contracted boxes. In the worst case, every one of the $m^n$ sub-boxes yields a nonempty contracted box, and hence the space complexity is $O(m^n)$.
\end{proof}
In practice, the observed cost is often substantially smaller than these worst-case bounds because contractor pruning and early termination in the $N_{\mathrm{it}}$ loop reduce the number of boxes retained and the number of iterations performed per box.

The next two interval methods use derivative information to accelerate contraction and verification. This typically improves practical performance but increases per-box cost because it requires interval Jacobian evaluation and (directly or indirectly) matrix inversion.

\subsection{Interval Linear Algebra: Determinant, Adjoint, and Inverse}
\begin{theorem}
Consider a continuously differentiable mapping $f : \mathbb{R}^n \to \mathbb{R}^n$ and an $n$-dimensional interval vector (box) $X \in \IR^n$. Suppose that (i) each scalar interval operation has maximum computational cost \(O(c)\), and (ii) each partial derivative $\frac{\partial f_i}{\partial x_j}$ can be evaluated on $X$ using at most $k$ scalar operations. Under these assumptions, the worst-case time complexity for computing the interval Jacobian
\(
    J(X) \in \IR^{n \times n}
\)
over $X$ is given by
\(
    T(n) = O(ckn^2).
\)
\end{theorem}

\begin{proof}
The interval Jacobian $J(X)$ has $n^2$ entries $J_{ij}(X)$, each obtained by evaluating the partial derivative $\frac{\partial f_i}{\partial x_j}$ on the interval box $X$ using interval arithmetic. By assumption (ii), evaluating one partial derivative uses at most $k$ scalar operations, and by assumption (i), each operation costs $O(c)$, so one entry $J_{ij}(X)$ costs $O(ck)$. 

With $n^2$ independent entries, the total cost of computing the interval Jacobian is
\begin{equation}\label{eqn:c_j}
    T(n) = n^2 \cdot O(ck) = O(ckn^2) = C_J.
\end{equation}
\end{proof}

Computing tight bounds for the determinant and the inverse of an interval matrix is NP-hard~\cite{rohn_handbook_2005, horacek_determinants_2018, horacek_interval_2016}. In practice one therefore computes \,\emph{enclosures} of $\det(A)$ and of the set of admissible inverses rather than exact ranges. For completeness---and to expose the worst-case cost drivers---we first analyze the classical Laplace/cofactor (adjugate) formulas for interval matrices and then contrast them with polynomial-time alternatives used in interval linear algebra.
\begin{theorem}\label{theorem:cc_determinant}
The computational time complexity of evaluating the determinant, \(\det(A)\), of an interval matrix \(A \in \IR^{n \times n}\) via Laplace expansion is \(O(10^{n} n!)\).
\end{theorem}
\begin{proof}
Laplace expansion along any fixed row expresses $\det(A)$ as a sum of $n$ terms
\[
\det(A)=\sum_{j=1}^{n}(-1)^{i+j} a_{ij}\,\det(M_{ij}),
\]
where each $M_{ij}$ is an $(n-1)\times(n-1)$ minor. Thus, computing $\det(A)$ via Laplace expansion requires computing $n$ determinants of size $(n-1)\times(n-1)$ and combining them with $n$ multiplications and $n-1$ additions.

Using the constant-factor bounds for interval arithmetic in Table~\ref{table:int_arith_cost}, each interval multiplication contributes only a fixed constant factor (here bounded by $10$) relative to a scalar real multiplication, and each interval addition contributes a fixed constant factor (bounded by $2$). Therefore the cost satisfies a recurrence of the form
\(
T(n) \le n\,T(n-1) + O(10n),
\)
where the $O(10n)$ term accounts for the $n$ interval multiplications and $n-1$ interval additions used to combine the $n$ minors.

Unrolling the recurrence yields
\(
T(n)=O\!\left(10^{n}\,n!\right),
\)
which we denote by
\begin{equation}\label{eqn:comp_cost_det}
    T(n)=O\!\left(10^{n} n!\right)=C_{det}.
\end{equation}
\end{proof}
\begin{theorem}\label{theorem:cc_adjoint}
The computational time complexity of determining the adjoint matrix, denoted by $\operatorname{adj}(A)$, for an interval matrix $A = \left\{\left[\underline{a_{ij}}, \overline{a_{ij}}\right]\right\} \in \IR^{n \times n}$ using the Laplace (cofactor) expansion method is $O(10^{\,n-1} \, n \cdot n!)$.
\end{theorem}
\begin{proof}
To determine the adjoint of a square matrix $A = \left\{ \left[\underline{a_{ij}}, \overline{a_{ij}}\right] \right\}$, it is necessary to compute $n^2$ cofactors. Consider the computation of a cofactor $CF_{ij}$, which involves obtaining the determinant of an $(n-1) \times (n-1)$ submatrix $A_{ij}$, while also applying the appropriate sign. 

From Eq.~\ref{eqn:comp_cost_det}, the complexity of computing $\det(A)$ is $O(10^n n!)$.
For each of the $n^2$ cofactors, calculating a determinant of size $(n-1) \times (n-1)$ incurs a cost of $O(10^{n-1} (n-1)!)$.
Consequently, the computational cost to calculate the adjoint is 
\begin{equation}\label{eqn:comp_cost_adj}
    T(n) = O(n^2 10^{n-1} (n-1)!) = O(10^{n-1} n \cdot n!) = C_{adj}.
\end{equation}
\end{proof}
\begin{theorem}\label{theorem:cc_inverse}
The computational time complexity \(C_{A^{-1}}\) associated with computing the inverse \(A^{-1}\) of an interval matrix \(A \in \IR^{n \times n}\) via the adjugate method is given by
\(
C_{A^{-1}} = O\!\bigl(10^{n-1}(10 + n)\, n! + 12 n^2\bigr).
\)
\end{theorem}
\begin{proof}
For a nonsingular $n \times n$ matrix $A$, the inverse can be written as $A^{-1} = \operatorname{adj}(A) / \det(A)$, where the adjugate matrix $\operatorname{adj}(A)$ is the transpose of the cofactor matrix.

Therefore, contributing factors in computational cost are $(i)$ the calculation of the determinant of $A$, $(ii)$ the computation of all $n^2$ cofactors, and $(iii)$ dividing each element of $\operatorname{adj}(A)$ by $\det(A)$.

From Eq.~\ref{eqn:comp_cost_det}, the computational cost of computing $\det(A)$ is $O(10^n n!)$. As given in Eq.~\ref{eqn:comp_cost_adj}, the cost of computing $\operatorname{adj}(A)$ is $O(10^{n-1} n \cdot n!)$. The computational cost of computing $\operatorname{adj}(A)/\det(A)$ is $O(12 n^2)$. Therefore, the computational cost of computing $A^{-1}$ is 
\begin{equation}
    T(n) = O(10^{n-1}(10 + n) n! + 12 n^2) = C_{A^{-1}}.
\end{equation}
\end{proof}

Theorems \ref{theorem:cc_determinant}-\ref{theorem:cc_inverse} establish the computational complexity of interval matrix operations when implemented via the Laplace (cofactor) expansion method, demonstrating that this approach yields extremely high, non-polynomial (exponential-factorial) complexity. 
These results indicate that, due to the practically infeasible nature of this recursive scheme, rigorous interval computations become computationally intractable; thus, they function primarily as theoretical worst-case complexity bounds rather than as a foundation for practical algorithmic implementations. 
In practical settings, conventional real-valued numerical techniques, such as Gaussian elimination and Krawczyk-based algorithms~\cite{moore_introduction_2009}, are generally employed for computing the determinants and inverses of interval matrices due to their relatively low computational complexity; see Table \ref{tab:interval_complexity}.  

The interval-valued matrix entries do not change the asymptotic dependence on the dimension $n$; they mainly affect the constant factors ($O(c)$) in the complexity bounds due to interval arithmetic operations and conceptually make it harder to obtain sharp interval enclosures.

\begin{table}[h!]
\centering
\caption{Time and space complexity of determinant and inverse computation for an $n \times n$ interval matrix}
\label{tab:interval_complexity}
\begin{tabular}{|c|c|c|c|c|}
\hline
\rowcolor{lightgray!25}
\textbf{Method} & \multicolumn{2}{c|}{\textbf{Determinant}} & \multicolumn{2}{c|}{\textbf{Inverse}} \\
\cline{2-5}
\rowcolor{lightgray!25}
 & \textbf{Time } & \textbf{Space} & \textbf{Time} & \textbf{Space} \\
\hline
Gaussian Elimination & $O(cn^3)$ & $O(cn^2)$ & $O(cn^3)$ & $O(cn^2)$ \\
\hline
Krawczyk-Based & $O(cn^3)$ & $O(cn^2)$ & $O(cn^3)$ & $O(cn^2)$ \\
\hline
\end{tabular}
\end{table}

For the reasons discussed above, in the subsequent interval-based methods, we employ the computational complexity results for interval matrix operations obtained via the Gaussian elimination technique.

\subsection{Interval Newton Method}
The Interval Newton Method is an interval version of the classical Newton--Raphson method, making it particularly advantageous in cases of parametric uncertainty or when multiple roots need to be found within a specified interval~\cite{chorasiya_quantitative_2023}. Algorithm~\ref{algo:interval_newton} summarizes the method.

\begin{algorithm}[h!]
    \caption{Interval Newton Method}
    \label{algo:interval_newton}
    \DontPrintSemicolon

    \KwIn{$f$, $X_0$, $U$, $\varepsilon$, $N_{\mathrm{it}}$}
    \KwOut{Set}

    \SetKwFunction{IntervalNewton}{IntervalNewton}
    \SetKwProg{Fn}{Function}{}{}

    \Fn{\IntervalNewton{$f$, $X$, $U$, $\varepsilon$, $N_{\mathrm{it}}$}}{
        \If{$0 \in F(X, U)$}{
            $J \leftarrow \text{Jacobian}(f, X, U)$\;
            \If{$ 0 \notin \text{determinant}(J)$}{
                \For{$i = 1,\ldots, N_{\mathrm{it}}$}{
                    $N \leftarrow \text{mid}(X) - \text{inv}(J) \cdot f(\text{mid}(X), U)$\;
                    $X_{\text{new}} \leftarrow \text{intersection}(N, X)$\;
                    \If{$\operatorname{diam}(X) - \operatorname{diam}(X_{\text{new}}) < tolerance$}{
                        \textbf{break}\;
                    }
                    $X \leftarrow X_{\text{new}}$\;
                }
                \If{$X_{\text{new}}\neq \emptyset$}{
                    \text{Save $X_{\text{new}}$ in Set}\;
                }
            }
            \Else{
                \If{$\operatorname{diam}(X) > \varepsilon$}{
                    \text{ Split $X$ into $X_l$ and $X_r$}\;
                    \IntervalNewton{$f$, $X_l$, $U$, $\varepsilon$, $N_{\mathrm{it}}$}\;
                    \IntervalNewton{$f$, $X_r$, $U$, $\varepsilon$, $N_{\mathrm{it}}$}\;
                }
                \Else{\text{Save $X$ in Set}\;}
            }
        }
        \KwRet{Set}\;
        }
    \IntervalNewton{$f$, $X_0$, $U$, $\varepsilon$, $N_{\mathrm{it}}$}\;
\end{algorithm}

In multi-steady-state systems \(0 \in F_x(X)\), where \(F_x\) is the interval extension of \(df/dx\), extended interval arithmetic is used. 
The Interval Newton Method is characterised by rapid convergence attributable to its quadratic convergence rate.
\subsubsection{Time Complexity}
In cases where multiple roots exist within the search space or the contraction is deemed inadequate, it becomes necessary to subdivide the search domain. 
\begin{theorem}\label{thm:inm_tc}
Let $N_{\mathrm{it}}$ denote the (maximum) number of interval Newton iterations performed per sub-box before termination 
. 
The worst-case computational time complexity of the interval Newton method (Algorithm~\ref{algo:interval_newton}) for computing multiple roots of the equation \(f(x,u) = 0\), where \(x \in X_0\) and \(u \in U\), is
\[
O\!\left(\frac{N_{\mathrm{it}}\, n^{3} \,\mathrm{Vol}(X_0)}{\varepsilon^{n}}\right).
\]
\end{theorem}
\begin{proof}
The maximum number of subdivisions is proportional to $\mathrm{Vol}(X_0)/\varepsilon ^ n$. For a fixed sub-box, each interval Newton iteration requires (i) evaluating the inclusion function (cost $C_F$), (ii) computing the interval Jacobian (cost $C_J$), and (iii) computing (an enclosure of) the inverse of the interval Jacobian (cost $C_{J^{-1}}$). Thus, performing at most $N_{\mathrm{it}}$ iterations per sub-box incurs cost $O\big(N_{\mathrm{it}}(C_F + C_J + C_{J^{-1}})\big)$. Consequently, the worst-case computational time complexity is
\begin{equation}
    \text{Time Complexity} = O\left(\frac{N_{\mathrm{it}}(C_F + C_J + C_{J^{-1}})\,\mathrm{Vol}(X_0)}{\varepsilon ^ n} \right).
\end{equation}

Using Eqs.~\ref{eqn:C_F} and \ref{eqn:c_j} and standard $O(cn^3)$ enclosure costs for interval-matrix inversion via Gaussian-elimination-type methods (Table~\ref{tab:interval_complexity}), we typically have $C_{J^{-1}} \gg C_J \gg C_F$, and hence
\begin{equation}
    O\left(\frac{N_{\mathrm{it}}(C_F + C_J + C_{J^{-1}})\,\mathrm{Vol}(X_0)}{\varepsilon ^ n} \right) \approx O\left(\frac{N_{\mathrm{it}}\, n^3\,\mathrm{Vol}(X_0)}{\varepsilon^n}\right).
\end{equation}
\end{proof}
\subsubsection{Space Complexity}
\begin{theorem}\label{thm:inm_sp}
The worst-case computational space complexity of the interval Newton method (Algorithm~\ref{algo:interval_newton}) for the simultaneous computation of multiple roots of the equation \(f(x,u) = 0\), where \(x \in X_0\) and \(u \in U\), is
\[
O\!\left( n^2 \,\log_2\!\left(\frac{\operatorname{diam}(X_0)}{\varepsilon}\right) + \frac{\mathrm{Vol}(X_0)}{\varepsilon^n} \right).
\]
\end{theorem}
\begin{proof}
There are two dominant memory contributions:
$(1)$ \emph{Stored boxes.} In the worst case, the method may need to retain up to \(O\!\left(\mathrm{Vol}(X_0)/\varepsilon^n\right)\) boxes (e.g., in a queue/worklist or in the output set).
$(2)$ \emph{Recursion stack and per-call Jacobians.} If the method bisects along a coordinate until widths are $\le\varepsilon$, the recursion depth is at most \(O\!\left(\log_2(\operatorname{diam}(X_0)/\varepsilon)\right)\). If each active call stores an interval Jacobian (an $n\times n$ interval matrix), this contributes \(O\!\left(n^2\log_2(\operatorname{diam}(X_0)/\varepsilon)\right)\) memory.
Summing these contributions yields the stated bound.
\end{proof}

\paragraph{Remark} The Interval Newton Method achieves quadratic contraction of search boxes, significantly reducing the number of subdivisions and computational effort for well-defined systems.
Despite its worst-case exponential complexity, the method exhibits practical efficiency when initial enclosures are favorable and the Jacobians are well-conditioned.

Within the framework of a system of equations, this method requires dividing by the interval enclosure of the Jacobian matrix; thus, when the system is ill-conditioned and the Jacobian matrix approaches singularity, convergence becomes problematic. 
The interval Krawczyk method addresses this problem differently: instead of working with an interval matrix, it uses the inverse of a real-valued midpoint matrix. This typically leads to better numerical stability, especially for ill-conditioned systems.

\subsection{Interval Krawczyk Method}
The Krawczyk method constitutes an alternative form of Newton's method, obviating the necessity for computing the inverse of an interval matrix~\cite{moore_introduction_2009,jaulin_applied_2001}.
In comparison to the Interval Newton Method, the Interval Krawczyk Method exhibits a reduced rate of convergence; however, it demonstrates enhanced robustness when applied to ill-conditioned multidimensional systems.
Algorithm~\ref{algo:interval_krawczyk} summarizes the method.

\begin{algorithm}[h!]
    \caption{Interval Krawczyk Method}
    \label{algo:interval_krawczyk}
    \DontPrintSemicolon

    \KwIn{$f$, $X_0$, $U$, $\varepsilon$, $N_{\mathrm{it}}$}
    \KwOut{Set}

    \SetKwFunction{IntervalKrawczyk}{IntervalKrawczyk}
    \SetKwProg{Fn}{Function}{}{}

    \Fn{\IntervalKrawczyk{$f$, $X$, $U$, $\varepsilon$, $N_{\mathrm{it}}$}}{
        \If{$0 \in F(X, U)$}{
            $J \leftarrow \text{Jacobian}(f, X, U)$\;
            \If{$ 0 \notin \text{determinant}(J)$}{
                \For{$i = 1,\ldots, N_{\mathrm{it}}$}{
                    $x_0 \leftarrow \text{mid}(X)$\;
                    $Y \leftarrow \text{inverse}(\text{mid}(J))$\;
                    $K \leftarrow x_0 - Y f(x_0,U) + (I - Y J) (X - x_0)$\;
                    $X_{\text{new}} \leftarrow \text{intersection}(K, X)$\;
                    \If{$\operatorname{diam}(X) - \operatorname{diam}(X_{\text{new}}) < tolerance$}{
                        \textbf{break}\;
                    }
                    $X \leftarrow X_{\text{new}}$\;
                }
                \If{$X_{\text{new}} \neq \emptyset$}{
                    \text{Save $X$ in Set}\;
                }
            }
            \Else{
                \If{$\operatorname{diam}(X) > \varepsilon$}{
                    \text{ Split $X$ into $X_l$ and $X_r$}\;
                    \IntervalKrawczyk{$f$, $X_l$, $U$, $\varepsilon$, $N_{\mathrm{it}}$}\;
                    \IntervalKrawczyk{$f$, $X_r$, $U$, $\varepsilon$, $N_{\mathrm{it}}$}\;
                }
                \Else{\text{Save $X$ in Set}\;}
            }
        }
        \KwRet{Set}\;
    }
    \IntervalKrawczyk{$f$, $X_0$, $U$, $\varepsilon$, $N_{\mathrm{it}}$}\;
\end{algorithm}

\subsubsection{Time Complexity}
In the presence of multiple roots, a bisection of the search space is employed. 
\begin{theorem}\label{thm:ikm_tc}
Let $N_{\mathrm{it}}$ denote the (maximum) number of Krawczyk iterations performed per sub-box before termination 
. The worst-case computational time complexity of the Krawczyk method (Algorithm~\ref{algo:interval_krawczyk}) employed to compute multiple roots of the equation \(f(x,u)=0\), with \(x \in X_0\) and \(u \in U\), is
\[
O\!\left(\frac{N_{\mathrm{it}}\, n^{3} \,\mathrm{Vol}(X_0)}{\varepsilon^{n}}\right).
\]
\end{theorem}
\begin{proof}
The proof follows a similar approach to the interval Newton method, with the distinction that the interval Krawczyk method inverts a real-valued (point) matrix rather than an interval-valued matrix. Inverting a point matrix via Gaussian elimination costs $O(n^3)$. Consequently, the time complexity is
\begin{equation}
    \text{Time Complexity} = O\left( \frac{N_{\mathrm{it}}(C_F + C_J + n^3)\,\mathrm{Vol}(X_0)}{\varepsilon ^ n} \right)
\end{equation}
As previously discussed, $O(n^3) \gg C_J \gg C_F$,
\begin{equation}
    \implies O\left( \frac{N_{\mathrm{it}}(C_F + C_J + n^3)\,\mathrm{Vol}(X_0)}{\varepsilon ^ n} \right) \approx O\left(\frac{N_{\mathrm{it}}\, n^3\,\mathrm{Vol}(X_0)}{\varepsilon^n} \right)
\end{equation}
\end{proof}
Two clarifications are worth noting. (i) While the adjugate/cofactor formula has factorial complexity, practical inversion of a real (point) matrix is performed using Gaussian-elimination-type methods with $O(n^3)$ time (and often with good constant factors in optimized libraries)~\cite{jeannerod_essentially_2005}. (ii) Although interval Newton and interval Krawczyk share the same leading-order worst-case scaling, Krawczyk is typically faster (per iteration) in practice because it inverts only the real-valued midpoint matrix $\mathrm{mid}(J)$ rather than an interval matrix.
\subsubsection{Space Complexity}
\begin{theorem}\label{thm:ikm_sp}
The worst-case computational space complexity of the Krawczyk method (Algorithm~\ref{algo:interval_krawczyk}) for computing multiple zeros of the equation \(f(x,u)=0\), with \(x \in X_0\) and \(u \in U\), is
\[
O\!\left( n^2 \,\log_2\!\left(\frac{\operatorname{diam}(X_0)}{\varepsilon}\right) + \frac{\mathrm{Vol}(X_0)}{\varepsilon^n} \right).
\]
\end{theorem}

\begin{proof}
The memory usage is dominated by (i) storing up to $O\!\left(\mathrm{Vol}(X_0)/\varepsilon^n\right)$ boxes in the worklist/output in the worst case, and (ii) storing $O\!\left(\log_2(\operatorname{diam}(X_0)/\varepsilon)\right)$ active calls in a recursion stack when subdivision is used, each potentially holding an $n\times n$ interval Jacobian (cost $O(n^2)$). Summing these terms gives the stated bound.
\end{proof}
\noindent

Table~\ref{tab:complexity_summary} summarizes the worst-case time and space complexity bounds derived above. For bisection/Newton/Krawczyk, the bounds are expressed in terms of the target tolerance $\varepsilon$, whereas subdivision$+$filter and constraint propagation are naturally parameterized by the fixed subdivision parameter $m$.

\begin{table}[t]
    \caption{Worst-case time and space complexities of the interval methods (upper bounds; parameterization depends on the method).}
    \label{tab:complexity_summary}
    \centering
    \small
    {\setlength{\tabcolsep}{2pt}
    \renewcommand{\arraystretch}{2.0}
    \renewcommand\tabularxcolumn[1]{m{#1}}
    \begin{tabularx}{\textwidth}{|>{\raggedright\arraybackslash}m{3.74cm}|>{\raggedright\arraybackslash}X|>{\raggedright\arraybackslash}m{6.5cm}|}
        \hline
        \rowcolor{lightgray!25}
        \textbf{Method} & \textbf{Time Complexity} & \textbf{Space Complexity} \\
        \hline
        Interval Bisection & $O\!\left(\dfrac{(C_F+n)\,\mathrm{Vol}(X_0)}{\varepsilon^n}\right)$ & $O\!\left(\dfrac{\mathrm{Vol}(X_0)}{\varepsilon^n}\right)$ \\
        \hline
        Subdivision $+$ Filter & $O\!\left(m^n\,C_F\right)$ & $O\!\left(m^n\right)$ \\
        \hline
        Constraint Propagation & $O\!\left(m^n\,N_{\mathrm{it}}\,(C_{\mathrm{con}}+n)\right)$ & $O\!\left(m^n\right)$ \\
        \hline
        Interval Newton & $O\!\left( N_{\mathrm{it}}\,(C_F+C_J+C_{J^{-1}})\,\dfrac{\mathrm{Vol}(X_0)}{\varepsilon^n}\right)$ & $O\!\left(\dfrac{\mathrm{Vol}(X_0)}{\varepsilon^n} + n^2\log_2\!\left(\dfrac{\operatorname{diam}(X_0)}{\varepsilon}\right)\right)$ \\
        \hline
        Interval Krawczyk & $O\!\left( N_{\mathrm{it}}\,(C_F+C_J+n^3)\,\dfrac{\mathrm{Vol}(X_0)}{\varepsilon^n}\right)$ & $O\!\left(\dfrac{\mathrm{Vol}(X_0)}{\varepsilon^n} + n^2\log_2\!\left(\dfrac{\operatorname{diam}(X_0)}{\varepsilon}\right)\right)$ \\
        \hline
    \end{tabularx}}
\end{table}

\subsection{Numerical Experiment}
\label{subsec:numerical_experiments}
\begingroup
\setlength{\floatsep}{2pt plus 1pt minus 1pt}
\setlength{\dblfloatsep}{2pt plus 1pt minus 1pt}
\setlength{\textfloatsep}{4pt plus 1pt minus 1pt}
\setlength{\dbltextfloatsep}{4pt plus 1pt minus 1pt}
\setlength{\intextsep}{4pt plus 1pt minus 1pt}
\setlength{\abovecaptionskip}{1pt}
\setlength{\belowcaptionskip}{1pt}
\captionsetup[table]{skip=1pt}

The complexity bounds in the preceding section are worst-case guarantees; we therefore complement them with numerical experiments that quantify observed workload and runtime, and help identify which validated primitives (interval evaluation, contraction, or interval linear algebra) dominate in representative instances.

\paragraph{Application context} Both test problems are steady-state computations for biochemical/synthetic-biology ODE models with parametric uncertainty. In these applications, steady states represent experimentally observable operating points (e.g., expression levels or RNA concentrations), and uncertainty intervals reflect cell-to-cell variability, measurement noise, and imperfectly identified kinetic constants. A validated enclosure of all steady states consistent with $u\in U$ supports application-driven tasks such as robust design (ensuring a circuit exhibits the intended qualitative behavior across parameter ranges), robust operating-point certification, and uncertainty-aware inference/identifiability screening. In particular, the output enclosure can be used to (i) certify that \emph{no} steady state lies in an undesirable range (safety/robustness screening) and (ii) quantify whether multiple operating points are possible under bounded uncertainty (multistability screening).

\paragraph{Computational scope} We report results for dimensions $n\in\{2,5,10\}$ for both examples (Tables~\ref{tab:runtime_vs_volume}, \ref{tab:ex1_n5}, \ref{tab:ex1_n10}, \ref{tab:ex2_i2}, \ref{tab:ex2_i5}, and \ref{tab:ex2_i10}). The higher-dimensional instances process up to $8.76\times 10^7$ boxes (Table~\ref{tab:ex2_i10}) and require up to $\approx 4.15\times 10^3$ seconds wall-clock time on a single workstation, providing substantial computational evidence alongside the worst-case complexity bounds.

All methods were implemented in Julia v1.11.0 using IntervalArithmetic v0.22.36 for interval evaluation and ReversePropagation v0.3.0 for contractor-based pruning. Experiments were executed on a single workstation in a single-process setting with a 2~GHz quad-core Intel Core~i5 CPU, 16~GB 3733~MHz LPDDR4X memory, and macOS~26.4.1 (25E253). Unless stated otherwise, each timing is reported as the median wall-clock time over 10 independent runs. All runs were performed in a single-threaded setting (one Julia thread), and we exclude one warm-up run to avoid just-in-time compilation effects in the reported timings.

We report the number of processed (visited) boxes $N_{\mathrm{proc}}$, the number of retained boxes $N_{\mathrm{keep}}$ in the final enclosure, and the median wall-clock time $T_{\mathrm{med}}$ in seconds. For iteration-based methods, Avg. Iter. denotes the average number of inner contractor/Newton/Krawczyk iterations per processed box.

Bisection/Newton/Krawczyk are parameterized by a tolerance $\varepsilon$ and terminate when each retained box has coordinate-wise width at most $\varepsilon$. Subdivision$+$filter uses a uniform discretization parameter $m$ (approximately $m^n$ boxes in $n$ dimensions) with no adaptive refinement beyond the chosen grid. Throughout, for fixed-grid methods we define $N_{\mathrm{proc}}$ to be the number of grid boxes processed (i.e., $N_{\mathrm{proc}}=m^n$); the initial domain $X_0$ is \emph{not} counted as an additional processed box. Constraint propagation uses an iteration limit $l$ (maximum number of contraction iterations per box). For iteration-based methods, an additional global cap $N_{\mathrm{it}}$ is enforced per processed box to prevent excessive inner-loop work.

We consider an uncertain steady-state problem of the form $f(x,u)=0$ with $x\in X_0$ and $u\in U$, where $U$ is an interval parameter set. The enclosure goal is to return a collection of boxes whose union contains all solutions in $X_0$ for all parameters $u\in U$. In this experiment, we employ a Hill-type regulatory-network steady-state model (commonly used in synthetic gene circuit modeling) following~\cite{gardner_construction_2000}, defined as follows:
\[
    f(x,u)=\begin{bmatrix}
        0.5 + \dfrac{\alpha_1}{1+x_n^{h}} - \gamma x_1\\
        0.5 + \dfrac{\alpha_2}{1+x_1^{h}} - \gamma x_2\\
        \vdots\\
        0.5 + \dfrac{\alpha_n}{1+x_{n-1}^{h}} - \gamma x_n
    \end{bmatrix},
\]
where $h$ is the Hill coefficient (in our experiments, we set $h=10$).
\[
    x=(x_1,\dots,x_n)\in X_0=[0,10]^n,\quad
    u=(\alpha_1,\dots,\alpha_n,\gamma)\in U=[3.8,4.2]^n\times[0.95,1.05].
\]
As a second numerical example, we consider a biochemical WTA network employing \emph{in vitro} transcriptional circuits as delineated in~\cite{kim_neural_2004}. The network consists of DNA switches, RNA polymerase (RNAP), and RNase. Enzyme saturation governed by Michaelis--Menten kinetics is used to induce global inhibition, as opposed to traditional neuron-like inhibition facilitated through explicit interconnections. Each DNA switch competes for the limited shared resource RNAP: a switch in the ON state sequesters more RNAP than one in the OFF state, thereby attenuating the activation of competing switches. This ultimately results in a WTA dynamic without the necessity for mutual inhibitory connections. The mathematical model is given by
\begin{equation}\label{eqn:wta_model}
\frac{d x_i}{dt} = \frac{E_{\text{tot}}}{1 + L} \left(
\frac{k_{\text{cat}}}{K_M} [D_{ii} x_i]
+ \frac{k'_{\text{cat}}}{K'_M} [D_{ii}]
\right) - \frac{E_{\text{tot},d}}{1 + L_d} \cdot \frac{k_{d,\text{cat}}}{K_{d,M}} x_i.
\end{equation}
Here, $x_i$ refers to the concentration of activator RNA for node $i$. $E_{\text{tot}}$ and $E_{\text{tot},d}$ are the total concentrations of RNA polymerase and RNase, respectively, and $L$ and $L_d$ denote the corresponding loads. The parameters $k_{\text{cat}},\,k'_{\text{cat}},\,k_{d,\text{cat}}$ are catalytic rate constants and $K_M,\,K'_M,\,K_{d,M}$ are Michaelis constants for the corresponding enzyme--substrate pairs. The quantities $[D_{ii}x_i]$ and $[D_{ii}]$ denote the concentrations of the DNA switch templates in the ON and OFF states, respectively.

The uncertain parameters are taken as
\(
D_{\mathrm{tot}} \in [1.98,\, 2.02], \quad
E_{\mathrm{tot}} \in [0.099,\, 0.101], \\ \quad 
K_M \in [0.0099,\, 0.0101], \quad
K'_M \in [0.099,\, 0.101], \quad
k_{\mathrm{cat}} \in [0.0297,\, 0.0303], \quad
k'_{\mathrm{cat}} \in [0.01188,\, 0.01212], \quad \\
k_{d,\mathrm{eff}} \in [0.00198,\, 0.00202], \quad
K_A \in [0.99,\, 1.01].
\)

\paragraph{How the examples differ} Example~1 (Hill network) is a smooth algebraic steady-state map with rational nonlinearities and a sparse ring structure (each component depends on one upstream state), with uncertainty in the production/degradation parameters $(\alpha_i,\gamma)$. It represents a prototypical uncertain regulatory steady-state problem for certifying feasible operating ranges (and possible multistability) under bounded parameter variation. Example~2 (WTA transcriptional circuit), derived from Michaelis--Menten saturation and global resource loading, adds coupling via load terms and sharper nonlinearities from enzyme saturation. It models competition for a shared resource and tests whether winner-take-all behavior (multiple equilibria with different ``winners'') persists under parametric uncertainty. Thus, Example~1 mainly probes how inclusion-function evaluation and Jacobian-based contraction scale with $n$, while Example~2 probes how stronger coupling/saturation affects pruning and the relative cost of contractor and interval linear-algebra primitives.

For interval bisection, we terminate when every retained sub-box has coordinate-wise width at most $\varepsilon$. For subdivision$+$filter, we fix the subdivision parameter $m$ (i.e., $m$ parts per dimension) and process the resulting set of sub-boxes, retaining only those boxes $X$ for which $0\in F(X,U)$; there is no adaptive termination beyond the chosen $m$. For iteration-based methods (interval constraint propagation, interval Newton, and Krawczyk), we cap the number of iterations per processed sub-box by $N_{\mathrm{it}}$ and terminate early on a sub-box if the decrease in its maximum diameter between successive iterations falls below the tolerance.
\begin{table}[h!]
\caption{Example 1 ($n=2$): Numerical experiment results as a function of the initial search region. We use $X_0=V_1=[0,10]^2$ and $X_0=V_2=[0,20]^2$, with $\mathrm{Vol}(V_2)=4\,\mathrm{Vol}(V_1)$.}
\label{tab:runtime_vs_volume}
\centering
{\setlength{\tabcolsep}{1.7pt}
\renewcommand{\arraystretch}{0.95}
\scriptsize
\renewcommand\tabularxcolumn[1]{>{\centering\arraybackslash}m{#1}}
\begin{tabularx}{\textwidth}{|>{\centering\arraybackslash}m{1.8cm}|>{\centering\arraybackslash}m{1.55cm}|*{8}{X|}}
\hline
\rowcolor{lightgray!25}
\textbf{Method} & \textbf{Setting} & \multicolumn{4}{c|}{$\boldsymbol{V_1} = [0, 10]^2$} & \multicolumn{4}{c|}{$\boldsymbol{V_2}=[0, 20]^2$}\\
\hline
\rowcolor{lightgray!10}
& & $\mathbf{N_{\mathrm{proc}}}$ & $\mathbf{N_{\mathrm{keep}}}$ & \textbf{Avg. Iter.} & \textbf{Time (s)} & $\mathbf{N_{\mathrm{proc}}}$ & $\mathbf{N_{\mathrm{keep}}}$ & \textbf{Avg. Iter.} & \textbf{Time (s)}\\
\hline
Bisection & $\varepsilon=10^{-3}$ & 485497 & 235585 & -- & 4.7 & 494709 & 240289 & -- & 4.8 \\ \hline
Subdivision $+$ Filter & $m=100$ & 10000 & 43 & -- & 0.095 & 10000 & 13 & -- & 0.095 \\ \hline
Constraint Propagation & $m=50$ & 2500 & 11 & 1.02 & 0.033 & 2500 & 7 & 1.02 & 0.035 \\ \hline
Newton & $\varepsilon=10^{-3}$ & 103 & 5 & 2.22 & 0.003 & 119 & 7 & 2.60 & 0.0035 \\ \hline
Krawczyk & $\varepsilon=10^{-3}$ & 103 & 5 & 4.22 & 0.004 & 119 & 7 & 4.60 & 0.0045 \\ \hline
\end{tabularx}}
\end{table}
Table~\ref{tab:runtime_vs_volume} provides a compact comparison of how workload and runtime change as the initial search region increases from $V_1=[0,10]^2$ to $V_2=[0,20]^2$ (a fourfold increase in area).

For the fixed-grid methods, the processed-box count is essentially set by the discretization parameter: subdivision$+$filter processes $N_{\mathrm{proc}}=m^2=10{,}000$ boxes for both domains and its runtime remains nearly constant ($\approx 0.095$~s), while constraint propagation processes $m^2=2{,}500$ boxes for both domains and exhibits only a small runtime increase ($0.033$~s to $0.035$~s). The retained-box counts, however, decrease noticeably as $X_0$ is enlarged (e.g., 43 to 13 for subdivision$+$filter), indicating that most of the additional region is excluded at coarse resolution.

For the adaptive methods, the effect of the larger search region is also modest in this instance. Bisection (tolerance $\varepsilon=10^{-3}$) shows only a slight increase in processed and retained boxes (485{,}497 to 494{,}709; 235{,}585 to 240{,}289) and a correspondingly small runtime change (4.7~s to 4.8~s), suggesting that the problem structure confines the ``difficult'' region to a similar subset of both domains. Interval Newton and Krawczyk process more boxes at $V_2$ (103 to 119) and require slightly more iterations on average (Newton: 2.22 to 2.60; Krawczyk: 4.22 to 4.60), yielding small runtime increases (0.0030~s to 0.0035~s; 0.0040~s to 0.0045~s). Notably, Newton achieves similar box counts with fewer iterations than Krawczyk in this example, consistent with the different contraction mechanisms.

Overall, Table~\ref{tab:runtime_vs_volume} reinforces two points that complement the worst-case bounds: (i) for fixed-grid variants, the choice of $m$ largely determines the workload (hence weak sensitivity to $\mathrm{Vol}(X_0)$ when $m$ is held fixed), and (ii) contraction/verification can keep the incremental cost of a larger initial search region small when large portions of the domain are certified infeasible early. Importantly, the theoretical expressions in Table~\ref{tab:complexity_summary} are worst-case upper bounds that assume refinement proceeds broadly over the search domain until the termination criterion is met; they therefore scale with $\mathrm{Vol}(X_0)$ as a conservative guarantee. In typical instances, however, the observed runtime is governed by an instance-dependent \emph{effective volume}--the subset of $X_0$ on which the inclusion test or contractor fails to quickly certify infeasibility. Here, by ``easy'' region we mean boxes for which infeasibility can be certified at coarse resolution (e.g., $0\notin F(X,U)$) or that are rapidly contracted, so that little or no further subdivision/iteration is required. When enlarging $X_0$ adds mostly such region, $N_{\mathrm{proc}}$ and runtime may increase sublinearly with $\mathrm{Vol}(X_0)$, as seen here despite a fourfold increase in area from $V_1$ to $V_2$.



\begin{table}[h!]
\caption{Example 1 ($n=5$): Numerical experiment results ($X_0=[0,10]^5$).}
\label{tab:ex1_n5}
\centering
{\setlength{\tabcolsep}{1.7pt}
\renewcommand{\arraystretch}{0.95}
\scriptsize
\renewcommand\tabularxcolumn[1]{>{\centering\arraybackslash}m{#1}}
\begin{tabularx}{\textwidth}{|>{\centering\arraybackslash}m{1.8cm}|>{\centering\arraybackslash}m{1.55cm}|*{4}{X|}}
\hline
\rowcolor{lightgray!25}
\textbf{Method} & \textbf{Setting} & $\mathbf{N_{\mathrm{proc}}}$ & $\mathbf{N_{\mathrm{keep}}}$ & \textbf{Avg. Iter.} & \textbf{Time (s)}\\
\hline
Bisection & $\varepsilon=10^{-2}$ & 13771 & 3774 & -- & 0.4 \\ \hline
Subdivision $+$ Filter & $m=5$ & 3125 & 31 & -- & 0.063 \\ \hline
Constraint Propagation & $l=5$ & 3125 & 1 & 1 & 0.12 \\ \hline
Newton & $\varepsilon=10^{-2}$ & 1361 & 1 & 1.33 & 0.21 \\ \hline
Krawczyk & $\varepsilon=10^{-2}$ & 1361 & 1 & 9.4 & 0.42 \\ \hline
\end{tabularx}}
\end{table}

\begin{table}[h!]
\caption{Example 1 ($n=10$): Numerical experiment results ($X_0=[0,10]^{10}$).}
\label{tab:ex1_n10}
\centering
{\setlength{\tabcolsep}{1.7pt}
\renewcommand{\arraystretch}{0.95}
\scriptsize
\renewcommand\tabularxcolumn[1]{>{\centering\arraybackslash}m{#1}}
\begin{tabularx}{\textwidth}{|>{\centering\arraybackslash}m{1.8cm}|>{\centering\arraybackslash}m{1.55cm}|*{4}{X|}}
\hline
\rowcolor{lightgray!25}
\textbf{Method} & \textbf{Setting} & $\mathbf{N_{\mathrm{proc}}}$ & $\mathbf{N_{\mathrm{keep}}}$ & \textbf{Avg. Iter.} & \textbf{Time (s)}\\
\hline
Bisection & $\varepsilon=10^{-1}$ & 5749985 & 322103 & -- & 1476.6551 \\ \hline
Subdivision $+$ Filter & $m=5$ & 9765625 & 1025 & -- & 340.5 \\ \hline
Constraint Propagation & $l=5$ & 9765625 & 3 & 1.00 & 729 \\ \hline
Newton & $\varepsilon=10^{-2}$ & 330277 & 50 & 1.075 & 102.62 \\ \hline
Krawczyk & $\varepsilon=10^{-2}$ & 330277 & 3 & 2 & 109.5 \\ \hline
\end{tabularx}}
\end{table}

\begin{table}[h!]
\caption{Example 2 ($i=\{1,2\}$): Numerical experiment results ($X_0=[0,2]^2$). For bisection/Newton/Krawczyk, we set $\varepsilon=10^{-2}\,\min\bigl(\mathrm{diam}(X_0)\bigr)$.}
\label{tab:ex2_i2}
\centering
{\setlength{\tabcolsep}{1.7pt}
\renewcommand{\arraystretch}{0.95}
\scriptsize
\renewcommand\tabularxcolumn[1]{>{\centering\arraybackslash}m{#1}}
\begin{tabularx}{\textwidth}{|>{\centering\arraybackslash}m{1.8cm}|>{\centering\arraybackslash}m{1.55cm}|*{4}{X|}}
\hline
\rowcolor{lightgray!25}
\textbf{Method} & \textbf{Setting} & $\mathbf{N_{\mathrm{proc}}}$ & $\mathbf{N_{\mathrm{keep}}}$ & \textbf{Avg. Iter.} & \textbf{Time (s)}\\
\hline
Bisection & $\varepsilon$ & 123 & 9 & -- & 0.0005 \\ \hline
Subdivision $+$ Filter & $m=10$ & 100 & 9 & -- & 0.0004 \\ \hline
Constraint Propagation & $l=5$ & 25 & 3 & 1.12 & 0.0003 \\ \hline
Newton & $\varepsilon$ & 65 & 4 & 1 & 0.0010 \\ \hline
Krawczyk & $\varepsilon$ & 65 & 4 & 7.75 & 0.0026 \\ \hline
\end{tabularx}}
\end{table}

\begin{table}[h!]
\caption{Example 2 ($i=\{1,\dots,5\}$): Numerical experiment results ($X_0=[0,2]^5$). For bisection/Newton/Krawczyk, we set $\varepsilon=10^{-2}\,\min\bigl(\mathrm{diam}(X_0)\bigr)$.}
\label{tab:ex2_i5}
\centering
{\setlength{\tabcolsep}{1.7pt}
\renewcommand{\arraystretch}{0.95}
\scriptsize
\renewcommand\tabularxcolumn[1]{>{\centering\arraybackslash}m{#1}}
\begin{tabularx}{\textwidth}{|>{\centering\arraybackslash}m{1.8cm}|>{\centering\arraybackslash}m{1.55cm}|*{4}{X|}}
\hline
\rowcolor{lightgray!25}
\textbf{Method} & \textbf{Setting} & $\mathbf{N_{\mathrm{proc}}}$ & $\mathbf{N_{\mathrm{keep}}}$ & \textbf{Avg. Iter.} & \textbf{Time (s)}\\
\hline
Bisection & $\varepsilon$ & 19458141 & 7882977 & -- & 137.8984 \\ \hline
Subdivision $+$ Filter & $m=10$ & 100000 & 323 & -- & 0.7166 \\ \hline
Constraint Propagation & $l=10$ & 100000 & 243 & 1 & 35.8382 \\ \hline
Newton & $\varepsilon$ & 19457249 & 7882977 & 1 & 523.6390 \\ \hline
Krawczyk & $\varepsilon$ & 19457249 & 7882977 & 96.2 & 449.2722 \\ \hline
\end{tabularx}}
\end{table}
\begin{table}[h!]
\caption{Example 2 ($i=\{1,\dots,10\}$): Numerical experiment results ($X_0=[0,2]^{10}$). For bisection/Newton/Krawczyk, we set $\varepsilon=10^{-1}\,\min\bigl(\mathrm{diam}(X_0)\bigr)$. For Newton/Krawczyk, Avg.~Iter.=0 indicates that the method immediately fell back to subdivision (e.g., due to failure of the Jacobian-based verification/preconditioning test) and therefore performed no inner iterations on the processed boxes.}
\label{tab:ex2_i10}
\centering
{\setlength{\tabcolsep}{1.7pt}
\renewcommand{\arraystretch}{0.95}
\scriptsize
\renewcommand\tabularxcolumn[1]{>{\centering\arraybackslash}m{#1}}
\begin{tabularx}{\textwidth}{|>{\centering\arraybackslash}m{1.8cm}|>{\centering\arraybackslash}m{1.55cm}|*{4}{X|}}
\hline
\rowcolor{lightgray!25}
\textbf{Method} & \textbf{Setting} & $\mathbf{N_{\mathrm{proc}}}$ & $\mathbf{N_{\mathrm{keep}}}$ & \textbf{Avg. Iter.} & \textbf{Time (s)}\\
\hline
Bisection & $\varepsilon$ & 87640737 & 25887072 & -- & 1342.3153 \\ \hline
Subdivision $+$ Filter & $m=5$ & 9765625 & 10449 & -- & 87.3882 \\ \hline
Constraint Propagation & $l=5$ & 9765625 & 1024 & 1 & 638.7728 \\ \hline
Newton & $\varepsilon$ & 87640737 & 25887072 & 0 & 4043.2741 \\ \hline
Krawczyk & $\varepsilon$ & 87640737 & 25887072 & 0 & 4149.7953 \\ \hline
\end{tabularx}}
\end{table}

\subsubsection{Comparative discussion}
\label{subsubsec:derived_metrics}
We interpret the results in Tables~\ref{tab:runtime_vs_volume}, \ref{tab:ex1_n5}, \ref{tab:ex1_n10}, \ref{tab:ex2_i2}, \ref{tab:ex2_i5}, and \ref{tab:ex2_i10} through two complementary lenses: (i) the total workload, as measured by $N_{\mathrm{proc}}$ and $N_{\mathrm{keep}}$, and (ii) the average cost per processed box, estimated a posteriori by $T_{\mathrm{med}}/N_{\mathrm{proc}}$. A closely related diagnostic is the keep-rate $N_{\mathrm{keep}}/N_{\mathrm{proc}}$, which indicates how readily a method can certify infeasibility or contract boxes before further subdivision.

Subdivision$+$filter exhibits the characteristic behavior of fixed-grid methods: for a fixed $m$, the workload is essentially determined by the grid size $m^n$ (up to minor bookkeeping), and hence is largely insensitive to the size of $X_0$. This is clearly visible in Table~\ref{tab:runtime_vs_volume}, where subdivision$+$filter and constraint propagation process the same number of boxes for $V_1=[0,10]^2$ and $V_2=[0,20]^2$, and the corresponding runtimes change only marginally. At the same time, their retained-box counts decrease as $X_0$ is enlarged, indicating that the additional region is largely excluded at coarse resolution (i.e., the effective feasible region does not expand commensurately with $\mathrm{Vol}(X_0)$).

In the $5$D instance of Example~1 (Table~\ref{tab:ex1_n5}), the adaptive methods already reduce the processed-box workload relative to uniform subdivision: Newton/Krawczyk process $N_{\mathrm{proc}}=1{,}361$ boxes (with $N_{\mathrm{keep}}=1$) versus $N_{\mathrm{proc}}=3{,}125$ for subdivision $+$ filter ($N_{\mathrm{keep}}=31$) and $N_{\mathrm{proc}}=3{,}125$ for constraint propagation ($N_{\mathrm{keep}}=1$). However, the per-box overhead is visible at this dimension: Newton/Krawczyk are slower than the fixed-grid subdivision$+$filter run (0.21--0.42~s versus 0.063~s) despite the smaller $N_{\mathrm{proc}}$.

In the $10$D instance of Example~1 (Table~\ref{tab:ex1_n10}), interval Newton and Krawczyk process substantially fewer boxes than bisection (e.g., $N_{\mathrm{proc}}=330{,}277$ versus $5{,}749{,}985$). Although each processed box is more expensive due to interval Jacobians and validated linear algebra, the reduction in $N_{\mathrm{proc}}$ dominates, yielding close to an order-of-magnitude reduction in wall-clock time relative to bisection.

In Example~2 with $i=\{1,\dots,5\}$ (Table~\ref{tab:ex2_i5}), Newton and Krawczyk retain the same number of boxes as bisection ($N_{\mathrm{keep}}=7{,}882{,}977$) and process essentially the same workload ($N_{\mathrm{proc}}\approx 1.95\times 10^7$), indicating that the additional verification/contraction steps do not lead to materially earlier pruning under the chosen uncertainty and tolerances. The higher per-box overhead therefore translates directly into increased runtime. For this example and parameter regime, Newton/Krawczyk behave as ``bisection plus overhead,'' and the keep-rate remains high, reflecting an intrinsically harder exclusion problem.

Constraint propagation can reduce the number of processed boxes dramatically (e.g., Table~\ref{tab:ex2_i5}), but its total runtime depends on the per-box contractor cost and the number of contraction iterations required to reach a fixed point. Consequently, a smaller $N_{\mathrm{proc}}$ does not necessarily imply a smaller $T_{\mathrm{med}}$.
The reported results reflect the practical regimes considered in our study; however, not all methods are run at identical parameter settings (e.g., different $\varepsilon$ or $m$ across rows in Table~\ref{tab:ex1_n10}). The tables should therefore be read as indicative of relative workload and dominant cost drivers, rather than as a strict ``like-for-like'' comparison at a single accuracy target.

Across the experiments, the dominant determinant of runtime is whether a method can materially reduce $N_{\mathrm{proc}}$ by certifying infeasibility or contracting boxes before further subdivision. When pruning is strong (Example~1, $n=10$), contraction-based methods (Newton/Krawczyk) can be substantially faster than bisection despite higher per-box costs. Conversely, when pruning is weak (Example~2, $i=\{1,\dots,5\}$), simpler methods with lower per-box overhead are preferable. Fixed-grid variants provide predictable work budgets controlled by $m$, but their cost grows rapidly with dimension through the $m^n$ dependence. These patterns mirror the structure of the worst-case bounds derived above and motivate the practical recommendations summarized in the conclusions.

\endgroup
\section{Conclusions}
We developed an algorithm-level worst-case time and space complexity framework for validated interval methods applied to uncertain nonlinear systems, with explicit dependence on the initial search region size (e.g., $\mathrm{Vol}(X_0)$), target tolerance $\varepsilon$, and the cost of validated primitives (inclusion-function and Jacobian evaluations, and interval linear algebra). The resulting bounds provide a principled baseline for assessing scalability and for comparing subdivision-based and derivative-based strategies.

For subdivision-based methods, we established worst-case growth proportional to the number of boxes required to reach tolerance, which scales as $\mathrm{Vol}(X_0)/\varepsilon^n$ in the uniform subdivision model. In particular, recursive branch-and-bound schemes such as interval bisection incur exponential growth in time and memory as dimension increases. Subdivision$+$filter improves constant factors and reduces recursion overhead by using an explicit worklist traversal.
For derivative-based methods, we showed that interval Newton and interval Krawczyk have worst-case time complexity of order $O\big(N_{\mathrm{it}}\, n^3 \,\mathrm{Vol}(X_0)/\varepsilon^n\big)$, reflecting $n\times n$ matrix operations repeated for at most $N_{\mathrm{it}}$ iterations per box. Although they share the same leading-order dependence, Krawczyk typically has smaller constants because it replaces interval-matrix inversion by inversion of a real-valued midpoint matrix. We also highlighted that naive interval determinant/inverse computation via Laplace expansion is intractable, underscoring the importance of specialized interval linear algebra routines.

Our numerical experiments on two application-motivated biochemical steady-state models (a Hill-type regulatory network and an enzyme-saturation-based winner-take-all circuit) complement the worst-case analysis by quantifying observed workloads and runtimes in dimensions $n\in\{2,5,10\}$, including instances that process up to $8.76\times 10^7$ boxes. The results illustrate when contraction and verification reduce the effective search region (yielding large reductions in $N_{\mathrm{proc}}$) and when, for the considered examples and parameter regimes, derivative-based methods behave as ``bisection plus overhead'' because pruning is weak.

We emphasize that these results are worst-case guarantees for the specific algorithms and primitives analyzed and are not lower bounds for validated enclosure problems. Improved inclusion functions, preconditioning, adaptive subdivision rules, and more sophisticated interval linear algebra may reduce practical costs and, in some cases, change asymptotic behavior.

Several directions appear particularly promising: (i) incorporating adaptive subdivision and priority rules to better exploit contraction, (ii) hybrid solvers that combine contractors with Newton/Krawczyk steps while using subdivision$+$filter rather than recursion, (iii) sharper interval extensions (e.g., reduced dependency) to limit overestimation, and (iv) parallel implementations that preserve validation guarantees. We expect the framework and bounds presented here to support both the analysis and the design of scalable, validated solvers in high-dimensional, uncertainty-driven applications.
\appendix

\section*{Appendix: Operation counts for basic interval arithmetic}\label{app:interval_arithmetic}
In this appendix we derive elementary operation counts for the basic interval arithmetic operations used throughout the paper. Let $A=[\underline{a},\overline{a}]$ and $B=[\underline{b},\overline{b}]$ be closed intervals with real endpoints. We count arithmetic operations on real numbers (addition, subtraction, multiplication, division, and comparisons) and treat each such primitive as having constant cost.

\subsection{Addition and subtraction}
Interval addition and subtraction are defined componentwise:
\begin{align}
A+B &= [\underline{a}+\underline{b},\; \overline{a}+\overline{b}],\\
A-B &= [\underline{a}-\overline{b},\; \overline{a}-\underline{b}].
\end{align}
Each requires two real additions/subtractions, hence constant cost.

\subsection{Multiplication}
For multiplication, one forms the four endpoint products
\[
P=\{\underline{a}\underline{b},\;\underline{a}\overline{b},\;\overline{a}\underline{b},\;\overline{a}\overline{b}\}
\]
and sets
\[
A\times B = [\min\{P\},\; \max \{P\}].
\]
This requires four real multiplications and then computing $\min$ and $\max$ of four values (three comparisons each), for constant cost.

\subsection{Division}
Assume $0\notin B$. Division is defined via multiplication by the reciprocal interval,
\[
A\div B = A \times \left[\frac{1}{\overline{b}},\;\frac{1}{\underline{b}}\right].
\]
Computing the reciprocal interval requires two real divisions. The subsequent multiplication uses the procedure above with four real multiplications and constant-time min/max comparisons. Therefore, interval division has a constant cost when $0\notin B$. (If $0\in B$, extended interval arithmetic is required and the operation may return a union of intervals; we do not analyse that case here.)

\section*{Acknowledgments}
The authors express their sincere gratitude to Prof. Ragesh Jaiswal for his constructive feedback.

\begingroup
\scriptsize
\makeatletter
\@ifundefined{bibsep}{}{\setlength{\bibsep}{0pt}}
\makeatother
\bibliographystyle{ieeetr}
\bibliography{references}
\endgroup
\end{document}